\documentclass[11pt]{article}

\usepackage{algorithm}
\usepackage{algorithmic}

\usepackage{amsmath,amssymb,amsthm}
\newtheorem{theorem}{Theorem}

\usepackage[table]{xcolor}
\usepackage{soul}
\usepackage{caption}
\usepackage{enumitem}
\usepackage{placeins}

\usepackage{caption}

\usepackage{booktabs}
\usepackage{tabularx}
\usepackage{makecell}

\usepackage{xurl}

\usepackage[final]{acl}

\usepackage{times}
\usepackage{latexsym}

\usepackage[T1]{fontenc}

\usepackage[utf8]{inputenc}

\usepackage{microtype}

\usepackage{inconsolata}

\usepackage{graphicx}

\title{Privacy-Preserving Heterogeneous Multi-LLM Federated Inference for Cognitive Diagnosis}

\author{%
  Yagna Manasa Boyapati, \textbf{Chong Yu}, \textbf{Tianyu Jiang}, \textbf{Justin Zhan} \\
  Department of Computer Science \\
  University of Cincinnati \\
 \texttt{boyapaya@mail.uc.edu, yuc5@ucmail.uc.edu} \\
\texttt{tianyu.jiang@uc.edu, zhanjt@ucmail.uc.edu}
}

\begin{document}
\maketitle
\begin{abstract}

Significant challenges remain in AI-driven educational systems in balancing privacy preservation with accurate cognitive diagnosis. In order to overcome this, we propose a federated inference framework in which several commercial LLM APIs collaborate without requiring access to raw student data or proprietary model internals. Using multiple federated entities, such as LLaMA-3.3-70B, GPT-4o-mini, and Claude-3-Haiku, our framework builds upon the heterogeneous multi-LLM architecture. The predictions generated by these entities are combined with $\varepsilon$-local differential privacy by adding Laplace noise locally to each entity’s prediction output before aggregation, while residual-based aggregation mitigates model heterogeneity. Our approach is predicated on an honest-but-curious trust paradigm in which API providers are presumed not to abuse submitted queries, and our differential privacy mechanism shields the published diagnostic results from external inference. We conduct rigorous privacy–utility analysis showing strong privacy guarantees with minimal accuracy loss, and extensive real-world evaluations across three educational benchmarks confirm the framework’s practical usability and cross-domain generalizability.

\end{abstract}

\section{Introduction}

The growing application of artificial intelligence in education has created unparalleled opportunities for individualized and adaptive learning, but it also raises significant challenges related to data privacy and the accuracy of diagnostic feedback for end users—namely, students \cite{wang2022neuralcd, liu2021towards}. Cognitive Diagnosis Models (CDMs), which infer a student’s mastery of specific concepts or skills, are foundational components of intelligent tutoring systems and adaptive learning environments \cite{wang2020neural, tatsuoka1983rule}. Despite their promise, two critical challenges limit the widespread deployment of CDM-based systems.

Traditional CDM approaches rely on the centralized processing of sensitive student data, raising serious privacy concerns. When learning institutions aggregate detailed information about students’ learning patterns and trajectories on a central server, the risk of data leakage, unauthorized access, or misuse increases substantially \cite{stephanie2022trustworthy, hu2020personalized}. Regulatory pressures, such as GDPR and FERPA \citep{gdpr2016regulation, rights2014family}, have heightened awareness of these issues and underscore the need for privacy-preserving solutions that can support the full range of cognitive diagnosis tasks.

Although Federated Learning (FL) enables collaborative model training without centralizing raw data \cite{elhussein2025federated}, and Differential Privacy (DP) provides strong guarantees for protecting individual-level information \cite{xin2024survey, fu2026differentially}, no existing work integrates these techniques into LLMs to address privacy concerns in CDM-based systems. The current LLM-based cognitive diagnosis approaches typically rely on a single language model to evaluate students’ knowledge states. While large language models show great promise in educational applications \cite{dong2025knowledge, chen2025llm}, single-model approaches are limited in robustness, adaptability, and diagnostic reliability. Prior work highlights the strengths of GPT-4 \cite{openai2023gpt4}, Claude \cite{anthropic2024claude}, and LLaMA \cite{touvron2023llama} across mathematical reasoning, linguistic understanding, and pattern recognition \cite{wang2025user}. However, these studies evaluate models independently, without leveraging  collaborative inference across heterogeneous LLMs.

Furthermore, privacy-preserving use of LLMs in federated learning has only recently been explored \cite{chen2024integration, tran2025privacy}, and no existing work investigates multi-LLM collaboration for cognitive diagnosis under privacy constraints.

To address these challenges, we propose the first privacy-preserving heterogeneous multi-LLM federated inference framework for cognitive diagnosis. As illustrated in Figure~\ref{fig:architecture}, LLaMA-3.3-70B (via Groq API) \cite{groq2024}, GPT-4o-mini \cite{openai2023gpt4}, and Claude-3-Haiku \cite{anthropic2024claude} serve as federated participants that collaboratively infer student knowledge states. By leveraging the complementary strengths of each model, our heterogeneous design eliminates the need for centralized training while simultaneously improving privacy, scalability, and robustness in real-world educational deployments.

\begin{figure}[t]
\includegraphics[width=\columnwidth]{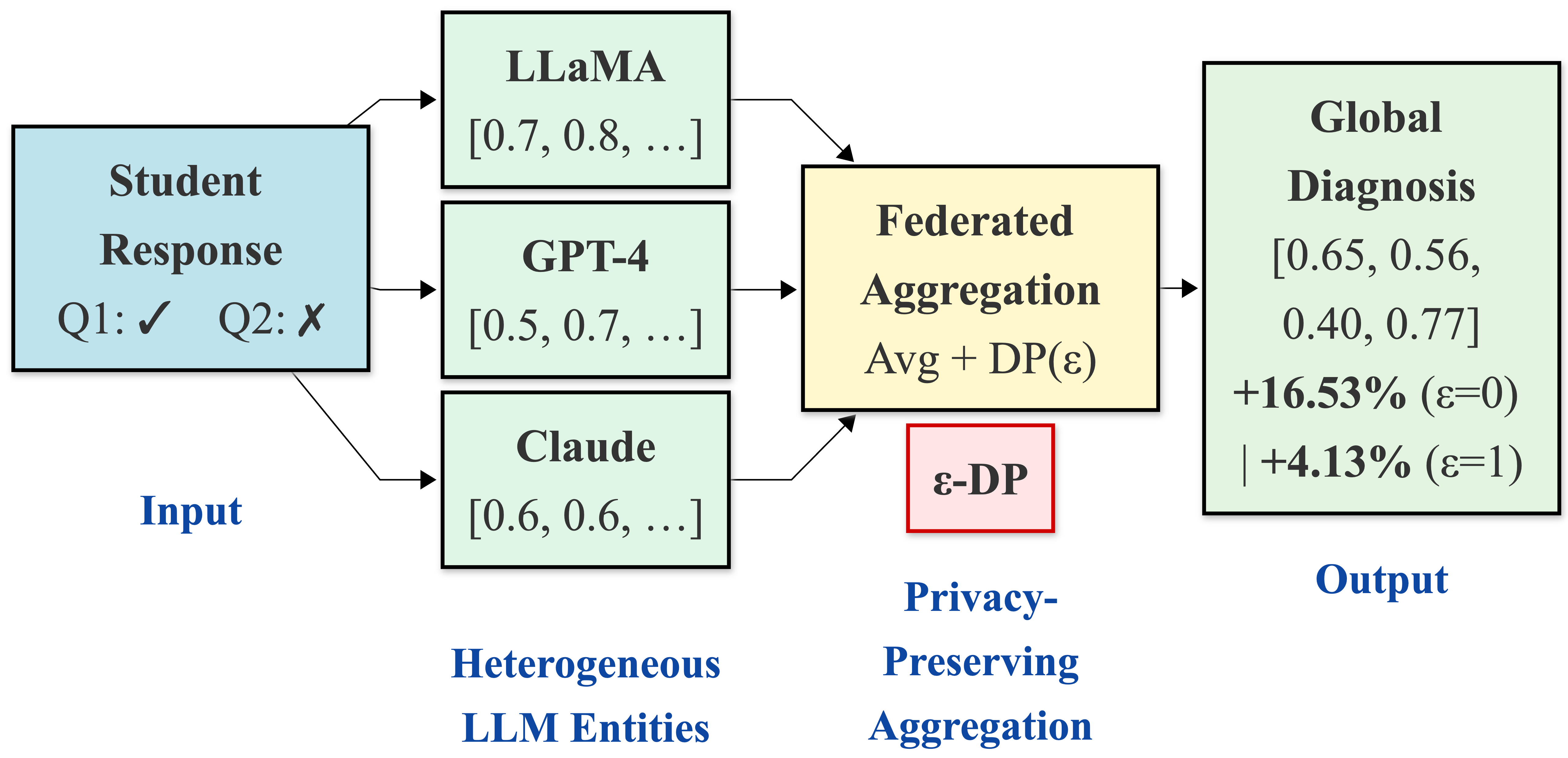}
\caption{Privacy-preserving heterogeneous multi-LLM federated framework with differentially private aggregation.}
\label{fig:architecture}
\end{figure}

We extend the paradigm of federated inference to heterogeneous multi-LLM architectures, where each model independently predicts knowledge states and contributes a complementary diagnostic perspective \cite{liu2024fault}. 
Unlike traditional FL methods, which aggregate model parameters or gradients and assume homogeneous models \cite{mcmahan2017communication, elhussein2025federated}, our framework operates entirely at the prediction level: LLM models are frozen and only accessible via black-box APIs, with only their output predictions shared and aggregated. This prediction-level, federated inference approach leverages model heterogeneity to yield more robust global cognitive assessments. Privacy-preserving aggregation uses Laplace mechanisms \cite{dwork2006calibrating, lin2023heterogeneous, liu2025differentially}, with calibrated noise applied locally at each entity's prediction output before aggregation, enabling institutions to use multiple commercial LLM APIs without exposing raw student data, a concern highlighted in blockchain-enabled FL studies \cite{jia2021blockchain}. We evaluate on three state-of-the-art benchmarks. Results show strong generalization across educational domains, with tight privacy guarantees and minimal utility loss, consistent with recent advances in DP-FL using Laplacian noise \cite{liang2024differentially} and heterogeneous DP mechanisms \cite{lin2023heterogeneous}.

It is critical to highlight that our privacy system uses an honest-but-curious trust model at the aggregate layer. Our $\varepsilon$-local differential privacy mechanism protects published knowledge-state diagnostic vectors from re-identification or membership inference by external parties accessing the outputs. It does not prevent data collection by commercial API providers (GPT-4o-mini, Claude-3-Haiku, and LLaMA-3.3-70B via Groq), which may log submitted queries in accordance with their respective policies. In Section 3.3, we describe mitigating options for deployments that require full end-to-end privacy, such as locally hosted open-source models and prompt data minimization.

Our key contributions are: (1) a novel, 
provider-agnostic federated inference framework 
coordinating commercial LLM APIs without access to 
proprietary model internals or training data; (2) a 
heterogeneous multi-LLM architecture treating 
LLaMA, GPT-4, and Claude as independent federated 
entities for cognitive diagnosis with locally 
protected student data; (3) rigorous privacy--utility 
analysis showing strong privacy preservation with 
minimal utility loss under $\varepsilon$-local 
differential privacy; and (4) extensive evaluation 
on three educational benchmarks validating practical 
implementability and cross-domain generalizability.
Our code and experimental resources are publicly available at \url{https://github.com/manasa2107/privacy-federated-llm-cognitive-diagnosis}.

\section{Related Work}
Our research integrates cognitive diagnosis models, federated learning, differential privacy, and large language models to advance educational applications.

\paragraph{Cognitive Diagnosis Models and LLM-based Diagnosis.}
Traditional CDM approaches, including Item Response Theory (IRT)~\citep{baker2001basics, yao2026a} and the Deterministic Input, Noisy-AND gate (DINA) model \citep{delatorre2009dina}, require substantial human intervention for feature engineering and struggle to capture complex skill relationships. Foundational knowledge tracing models including BKT \citep{corbett1994knowledge}, DKT \citep{piech2015deep}, and EKT \citep{liu2019ekt} improved sequential student modeling. Neural approaches such as NeuralCD \citep{wang2022neuralcd} and explainable CDM frameworks \citep{yang2022novel} further improved predictive accuracy via embedding-based interaction modeling on large-scale educational datasets. More recently, LLMs have been applied to cognitive diagnosis: \citet{dong2025knowledge} augment CDMs with LLM-generated diagnostic signals via contrastive learning and mask-reconstruction, while \citet{chen2025llm} evaluate GPT-4, Claude, and LLaMA individually across mathematical concepts demonstrating complementary strengths across models but without privacy or federation considerations. Beyond cognitive diagnosis, LLMs have shown broad promise across real-world educational tasks. \citet{sauberli-etal-2025-llms} evaluate the psychometric plausibility of LLM responses in standardized assessments using IRT and classical test theory, finding that LLM response distributions diverge from human patterns, motivating hybrid frameworks that combine LLM inference with structured diagnostic models. \citet{kim2025sqkt} propose knowledge tracing models that leverage LLM-extracted skill signals from student questions to improve performance prediction in programming education. \citet{chowdhury2025large} identify critical challenges in deploying LLMs for education, emphasizing that LLMs must align with pedagogical principles, integrate with knowledge tracing systems, and incorporate robust data privacy safeguards. \citet{vanzo-etal-2025-gpt} and \citet{luo2024chain} further demonstrate LLMs' growing role in real-world educational applications, including homework tutoring and question generation.

\paragraph{Federated Learning with Differential Privacy.}
Federated Learning enables collaborative model training without centralizing raw data \citep{elhussein2025federated}, making it well-suited for privacy-sensitive educational settings. Prior FL work has addressed statistical heterogeneity across participants \citep{mcmahan2017communication, li2020federated}, and applications in healthcare have combined FL with blockchain and differential privacy \citep{stephanie2022trustworthy, jia2021blockchain}, though often at significant computational overhead. Differential Privacy provides formal guarantees for protecting individual-level information via calibrated noise mechanisms \citep{dwork2006calibrating}. Foundational DP-FL works establish client-level privacy \citep{geyer2017differentially} and gradient-based noise mechanisms \citep{abadi2016deep}, which our framework adapts to the prediction-level aggregation setting. Recent work has explored DP in federated setups \citep{xin2024survey, fu2026differentially}, for LLM fine-tuning \citep{liu2025differentially}, and in personalized FL \citep{hu2020personalized, tran2025privacy}. Our solution leverages the Laplace mechanism \citep{liang2024differentially} with heterogeneous privacy budgets \citep{lin2023heterogeneous}. Most existing FL approaches assume homogeneous neural networks and aggregate model parameters or gradients. In contrast, we extend FL to heterogeneous commercial LLMs operating exclusively via black-box APIs, aggregating predictions at inference time with no access to model internals.
  
\begin{table}[t]
\centering
\scriptsize
\setlength{\tabcolsep}{2pt}
\renewcommand{\arraystretch}{1.1}
\begin{tabularx}{\columnwidth}{p{1.9cm}XXXXX}
\toprule
\textbf{Work} & \textbf{Het.} & \textbf{BB-API} & \textbf{Pred.} & \textbf{DP} & \textbf{CD} \\
\midrule
Chen et al. (2024) & $\times$ & $\times$ & $\times$ & $\times$ & $\times$ \\
Dong et al. (2025) & $\times$ & $\times$ & $\times$ & $\times$ & \checkmark \\
Chen et al. (2025) & \checkmark & \checkmark & $\times$ & $\times$ & \checkmark \\
Elhussein (2025)   & $\times$ & \checkmark & $\times$ & \checkmark & $\times$ \\
\midrule
\textbf{Our framework} & \checkmark & \checkmark & \checkmark & \checkmark & \checkmark \\
\bottomrule
\end{tabularx}
\caption{Comparison with related work. \checkmark~=~supported. Het.~=~Heterogeneous LLMs; BB-API~=~Black-Box API Access; Pred.~=~Prediction-level Aggregation.}
\label{tab:c1}
\end{table}

Unlike prior work relying on single LLMs with centralized data \citep{dong2025knowledge, chen2025llm}, homogeneous federated training \citep{elhussein2025federated, chen2024integration}, or LLM-education frameworks that lack formal privacy guarantees \citep{chowdhury2025large, kim2025sqkt, sauberli-etal-2025-llms}, our framework is the first to simultaneously combine heterogeneous commercial LLMs as federated inference entities, prediction-level aggregation without model weight access, residual correction for calibration bias, and local differential privacy and validated across three diverse educational benchmarks for cognitive diagnosis, as illustrated in Table~\ref{tab:c1}.

\section{Methodology}

Our privacy-preserving heterogeneous multi-LLM 
federated inference framework consists of four 
key steps: (1) Distributed LLM-augmented cognitive 
diagnosis, where each model independently predicts 
student knowledge states; (2) Privacy-preserving 
aggregation using local differential privacy to 
secure the combined predictions; (3) Residual 
correction for heterogeneous models to account 
for differences in model behavior and scale; and 
(4) Global diagnosis sharing, enabling all 
participants to benefit from the aggregated, 
privacy-protected diagnostic outputs. 
Figure~\ref{fig:pipeline} (Appendix~\hyperref[app:Algorithms]{E}) illustrates 
the overall system architecture and processing 
pipeline corresponding to these four stages.

\subsection{Problem Formulation}

Let $\mathcal{D} = \{(s_i, r_i, y_i)\}_{i=1}^N$ denote a student response dataset, where $s_i$ represents student $i$'s response vector, $r_i$ is the correctness indicator, $y_i \in \{0,1\}^K$ is the knowledge state vector for $K$ concepts, and $N$ represents total number of students. Let $\mathcal{S}$ denote the space of all possible student responses and $\mathcal{Q}$ 
denote the set of all questions. In cognitive diagnosis, we aim to learn a function $f: \mathcal{S} \times \mathcal{Q} \rightarrow [0,1]^K$ that predicts mastery probability for each knowledge concept given a student response to a particular question.

We examined a federated inference situation with $M$ heterogeneous LLM entities $\{\mathcal{M}_1, \mathcal{M}_2, \ldots, \mathcal{M}_M\}$. 
Two distinct phases must be clearly distinguished:

\paragraph{Training/Calibration Phase.} The dataset $\mathcal{D}$ is divided into $M$ disjoint subsets $\mathcal{D} = \bigcup_{j=1}^{M} \mathcal{D}_j$. Each subset $\mathcal{D}_j = \{(s_{i,j}, r_{i,j}, y_{i,j})\}_{i=1}^{N}$ is held solely by entity $\mathcal{M}_j$. 
This setup simulates the real-world scenario where multiple schools or institutions have separate student populations. Entities only access their local dataset and do not exchange raw student records with others or the aggregator. Our implementation has three processes: LLaMA-3.3-70B $\mathcal{D}_1$, GPT-4o-mini $\mathcal{D}_2$, and Claude-3-Haiku $\mathcal{D}_3$. Entities employ local data to calibrate their residual correction term, $r_j = \mathbb{E}[\hat{y}_j - \bar{y}]$, over their validation subset.
\paragraph{Inference Phase.} At test time, test 
student $s_t$ is evaluated simultaneously by all 
$M$ entities via their APIs. Each entity independently 
generates $\hat{y}_{t,j} \in [0,1]^K$, which are 
aggregated with local DP noise and residual correction 
to produce the global diagnosis. This does not violate 
disjoint partitioning: partitioning governs calibration 
data, while inference queries all entities on the same 
test student. The two phases serve different purposes 
on different data splits.

\subsection{Heterogeneous Multi-LLM Cognitive Diagnosis}

\paragraph{Local Assessment Phase.} All LLM entities operate on the analysis of student responses using engineered prompts that target cognitive diagnosis \cite{mcmahan2017communication}. For student response $s_i$ to question $q$, entity $j$ generates:
\begin{equation}
\hat{y}_{i,j} = M_j(\text{prompt}(s_{i,j}, q_{i,j}, \text{concepts})) \in [0,1]^K,
\end{equation}

where the prompt directs the LLM to assess the levels of mastery on $K$ concepts based on the quality of the student's response, reasoning, and conceptual thinking.

\paragraph{Q-Matrix Integration.} To associate questions with knowledge concepts, we use Q-Matrix theory \cite{tatsuoka1983rule, tatsuoka2009cognitive}. The Q-matrix $\mathbf{Q} \in \{0,1\}^{|\mathcal{Q}| \times K}$ is a binary matrix defining which knowledge concepts are required to solve each question:
\begin{equation}
Q_{qk} = \begin{cases}
1, & \text{if question } q \text{ assesses concept } k \\
0, & \text{otherwise}
\end{cases}
\end{equation}

The Q-matrix serves as a bridge between the observable student responses and the knowledge concepts in the latent space. For example, in ASSIST09 with concepts [Algebra, Geometry, Probability], a quadratic equation would have $Q_{q,\text{Algebra}}=1$ and $Q_{q,\text{Geometry}}=Q_{q,\text{Probability}}=0$. This constraint ensures that LLM assessments focus on relevant concepts for each question, enabling accurate knowledge state estimation and generalization across educational domains.

In Equation (1), when entity $M_j$ evaluates response $s_{i,j}$ to question $q$, the prompt includes only concepts where $Q_{qk}=1$. For instance, if $\mathbf{Q}_q = [1, 0, 1, 0]$, the LLM assesses only concepts 1 and 3, providing targeted diagnosis rather than blanket assessment across all $K$ concepts.

\subsection{Privacy-Preserving Federated Aggregation}

\paragraph{What We Protect.} 
Our privacy mechanism protects the published student knowledge-state diagnostic vectors against re-identification, membership inference, or leakage of sensitive information by external parties who access the diagnostic outputs. Our $(\varepsilon, 0)$-DP guarantees that the addition or subtraction of any individual student does not alter the output distribution by more than a factor of $e^\varepsilon$. 

\paragraph{Threat Model and Scope of Privacy Guarantees.} The scope of privacy guarantees and the threat model. The stated scope of our honest-but-curious trust paradigm is as follows:
(a) What our DP protects: as described above, published diagnostic vectors are protected against re-identification and membership inference by external parties.
(b) What is not protected by our DP: The commercial API providers (OpenAI, Anthropic, Groq) are still able to log and process the raw student answer requests that are sent to their endpoints notwithstanding our mechanism. Our DP guarantee does not cover API-level exposure; each LLM entity transmits student responses and Q-matrix prompts to its own API.
(c) Mitigation techniques for complete end-to-end privacy: Practitioners who need more stringent privacy can: (i) replace commercial APIs with locally hosted open-source models (like LLaMA via Ollama), which completely removes API-level exposure; (ii) apply input-side DP perturbation to student responses prior to API submission; or (iii) employ prompt data minimization, which sends only Q-matrix-relevant features instead of raw responses.


\paragraph{Differential Privacy Mechanism.} 
To protect individual student information from the centralized aggregator, we apply $\varepsilon$-differential privacy~\cite{dwork2006calibrating} at the entity level before aggregation. Each LLM entity $\mathcal{M}_j$ adds calibrated Laplace noise to its local predictions before sharing them with the aggregator, ensuring that the aggregator cannot infer individual student responses from the received predictions. Following the composition theorem for differential privacy~\cite{dwork2014algorithmic}, the aggregated result satisfies:

\begin{equation}
\begin{aligned}
\tilde{y}_i &= \frac{1}{M} \sum_{j=1}^{M} \left(\hat{y}_{i,j} + \boldsymbol{\eta}_j\right), \\
\eta_{jk} &\sim \mathrm{Lap}\!\left(\frac{\Delta f_j}{\varepsilon}\right),
\quad k = 1, \ldots, K.
\end{aligned}
\label{eq:aggregation}
\end{equation}


where $\mathrm{Lap}(\lambda)$ denotes the Laplace distribution with scale parameter $\lambda$, and $\boldsymbol{\eta}_j = [\eta_{j1}, \ldots, \eta_{jK}]$ represents a $K$ dimensional noise vector. The parameter $\Delta f_j$ denotes the local sensitivity, which bounds the $\ell_1$ norm of the prediction function, while $\epsilon$ represents the privacy budget. In our experiments, we evaluate $\epsilon \in \{1.0, 2.0\}$ to analyze the privacy utility tradeoff.




\paragraph{Sensitivity Analysis.}For cognitive diagnosis predictions of variables taking values in  $[0,1]^K$, the global sensitivity is:

\begin{equation}
\Delta f = \max_{y, y'} \|\tilde{y} - \tilde{y}'\|_1 \leq K
\end{equation}

We conservatively set $\Delta f = K$ to ensure $(\epsilon, 0)$-differential privacy~\cite{lin2023heterogeneous}.


\subsection{Aggregation with Residual Correction}

\paragraph{Heterogeneous Aggregation Challenge.}
Directly averaging predictions from heterogeneous LLMs introduces a fundamental inconsistency: the ``average of products $\neq$ product of averages'' problem documented in federated learning literature \cite{elhussein2025federated, zhou2025ensemble, dietterich2000ensemble}. Because each model operates with its own internal representation of concept mastery, na\"{i}ve aggregation fails to preserve the underlying probabilistic structure of cognitive diagnosis. In practice, heterogeneous LLMs differ substantially in (1) calibration (i.e., confidence and probability scaling), (2) granularity in assessing fine-grained skills or concepts, and (3) systematic biases arising from their distinct training corpora and architectural priors. These mismatches cause direct averaging to distort global knowledge estimates, motivating the need for residual-corrected and privacy-preserving aggregation tailored to heterogeneous LLM participants.


\paragraph{Residual Correction Mechanism.} Inspired by FedEx-LoRA techniques \cite{liu2025differentially}, we extend the model by incorporating residual corrections to handle heterogeneity. For each entity $j$, we compute residual terms:

\begin{equation}
r_j = \mathbb{E}[\hat{y}_j - \bar{y}],
\end{equation}

\noindent where $\bar{y}$ is the global mean prediction. The corrected aggregation becomes:

\begin{equation}
y_{\text{global}} = \frac{1}{M}\sum_{j=1}^M (\hat{y}_j - \alpha \cdot r_j),
\end{equation}

\noindent with $\alpha \in [0,1]$ as the correction strength. In practice, we set $\alpha = 0.3$ based on validation performance.

\subsection{Proposed Algorithms}

Algorithm~\ref{alg:federated_ldp} and Algorithm~\ref{alg:baseline} are provided in 
Appendix~\hyperref[app:Algorithms]{E} due to page constraints.
Algorithm~\ref{alg:federated_ldp} presents our complete federated cognitive diagnosis framework with privacy preservation. Algorithm~\ref{alg:baseline} presents the baseline approach without privacy mechanisms for comparison.

\subsection{Theoretical Privacy Guarantees}

Our mechanism satisfies $(\varepsilon, 0)$- local differential privacy:

\begin{theorem}[Privacy Guarantee]
Algorithm~1 with Laplace noise satisfies 
$\varepsilon$-local differential privacy for any 
adjacent datasets $D, D'$ differing in one student 
record.
\end{theorem}

\begin{proof}[Proof Sketch]
The Laplace mechanism with scale $\Delta f/\varepsilon$ 
provides $\varepsilon$-DP when applied to a function 
with sensitivity $\Delta f$ \citep{dwork2006calibrating}. 
Since predictions are bounded in $[0,1]^K$, the 
$\ell_1$ sensitivity is $\Delta f = K$, and 
post-processing invariance ensures residual 
correction preserves the guarantee 
\citep{dwork2014algorithmic}.
\emph{Detailed proof is provided in Appendix~\hyperref[app:proof]{F}.}
\renewcommand{\qedsymbol}{} 
\end{proof}

\paragraph{LDP Sensitivity Analysis.} A natural 
concern with LDP is that it requires stronger noise 
than central DP, implying higher utility loss. However, 
our bounded $[0,1]^K$ prediction outputs, already 
group-level knowledge-state forecasts rather than raw 
records, keep sensitivity low via three properties: 
(1) Bounded output space: predictions 
$\hat{y}_{i,j} \in [0,1]^K$ limit $\ell_1$ sensitivity 
to $K$ regardless of input dimensionality; 
(2) Group-level assessment: each entity's 
prediction reflects its local student body, giving 
effective per-student sensitivity $K/M \leq K/3$; 
(3) Per-concept noise: Laplace noise 
$\text{Lap}(K/\varepsilon)$ applied independently 
per concept with $[0,1]$ clipping further suppresses 
noise impact. Together, these explain our empirical 
privacy cost of just $0.40\%$ MAE at $\varepsilon=2.0$, 
significantly less than traditional high-dimensional 
LDP settings. We confirm $\varepsilon=2.0$ is within 
normal educational privacy levels 
\citep{dwork2006calibrating} via 
Rényi-DP accounting \citep{mironov2017renyi}.

\section{Results}

\subsection{Experimental Setup}
\paragraph{Datasets.} We evaluate on three diverse educational benchmarks: (1) ASSIST09 \cite{feng2009addressing, heffernan2014assistments}: Real-world mathematics dataset with 4,217 students, 26,688 responses across 124 concepts (we select top 4: Equations, Percentages, Integers, Conversions). \textit{Input:} Response patterns (correct/incorrect). \textit{Output:} Knowledge state $[0,1]^4$. \textit{Calculation:} Proportion of correct answers per concept using Q-matrix. 
(2) GSM8K \cite{cobbe2021training}: Grade school math word problems with 8,500 problems. \textit{Concepts (5):} Addition/Subtraction, Multiplication/Division, Fractions/Percentages, Multi-step Reasoning, Word Problems. \textit{Input:} Solution attempts. \textit{Output:} $[0,1]^5$. \textit{Calculation:} Keyword extraction, success rate per concept.
(3) UCI Student Performance \cite{cortez2008using}: Holistic assessment with 649 students, 33 attributes. \textit{Concepts (5):} Study Habits, Family Support, School Engagement, Social Factors, Academic Foundation. \textit{Input:} Demographics and behaviors. \textit{Output:} $[0,1]^5$. \textit{Calculation:} Attribute grouping, normalization, final grade G3. We examined our framework on the total student population for three educational standards. Each dataset contains 4-5 knowledge ideas, with an 80/20 train-test split.  

\paragraph{EdNet Scalability Analysis.} We also present findings on EdNet \cite{choi2020ednet}, a large-scale dataset with over 1.3 million student interactions from KnowledgeTag spanning 188 knowledge topics, to allay worries regarding small-scale evaluation. To show that our framework scales gracefully, we sample three subsets of increasing size (10K, 50K, and 100K interactions): MAE degrades by less than 2\% from the 10K to 100K setting, and processing time scales nearly linearly with student count, confirming practical deployability at scale beyond our primary benchmarks.

\paragraph{Q-Matrix Automatic Extraction.} Our 
framework does not strictly require expert-crafted 
Q-matrices. ASSIST09 uses a community-provided 
Q-matrix; GSM8K uses programmatic extraction via 
keyword matching and LLM-generated concept tags 
(GPT-4o-mini labels each problem with required 
concepts); UCI derives concepts through attribute 
grouping and normalization. Consistent improvements across all three construction  methods confirm that our framework generalizes well under automatically generated Q-matrices, without requiring expert annotation.

\paragraph{LLM Entities.} Our federation comprises three heterogeneous commercial LLMs: (1) LLaMA-3.3-70B via Groq API \cite{groq2024} (free tier, unlimited inference); (2) GPT-4o-mini via OpenAI API \cite{openai2023gpt4}; (3) Claude-3-Haiku via Anthropic API \cite{anthropic2024claude}.\footnote{Claude-3-Haiku was used for experiments conducted prior to April 2025. Due to model deprecation, subsequent experiments used \texttt{claude-sonnet-4-20250514}. Both model versions followed identical prompting protocols; results were verified for consistency across the transition.} The selection strategy follows the analysis made by Chen et al. in 2025 for LLM-CDM, which indicates the complementary advantages: LLaMA for mathematical reasoning, GPT-4 for conceptual understanding, Claude for nuanced evaluation.


\paragraph{Infrastructure.} Experiments were performed on the University of Cincinnati ARCC2 GPU cluster using 1 Tesla V100S GPU (32GB memory) per dataset, with all three datasets processed in parallel for efficient computation. Full-scale experiments across all three datasets and four privacy configurations ($\varepsilon \in \{0.5, 1.0, 2.0, \infty\}$) required approximately 24 total GPU-hours. Commercial LLM API requests were distributed across endpoints using load balancing with exponential backoff for rate limiting.

\paragraph{Hyperparameters.} (1) Privacy budgets: $\varepsilon \in \{0.5, 1.0, 2.0, \infty\}$ (where $\infty$ indicates no privacy). (2) Correction strength: $\alpha = 0.3$. (3) Aggregation: Simple averaging for baseline, residual-corrected for full model. (4) Temperature: 0.7 for all LLM APIs. (5) Max tokens: 512 per response.

\paragraph{Evaluation Metrics.} We evaluate using the following metrics:

(1) Mean Absolute Error (MAE):
\begin{equation}
\text{MAE} = \frac{1}{NK}\sum_{i=1}^N \sum_{k=1}^K |\tilde{y}_{ik} - y_{ik}|
\end{equation}

(2) Root Mean Square Error (RMSE):
\begin{equation}
\text{RMSE} = \sqrt{\frac{1}{NK}\sum_{i=1}^N \sum_{k=1}^K (\tilde{y}_{ik} - y_{ik})^2}
\end{equation}

(3) Improvement:
\begin{equation}
\text{Improvement} = \frac{\text{Baseline MAE} - \text{Federated MAE}}{\text{Baseline MAE}}
\end{equation}

(4) Privacy Cost:
\begin{equation}
\text{Privacy Cost} = \frac{\text{MAE}(\varepsilon) - \text{MAE}(\infty)}{\text{MAE}(\infty)}
\end{equation}
\noindent where $\text{MAE}(\varepsilon)$ denotes federated MAE at privacy budget $\varepsilon$, $\text{MAE}(\infty)$ denotes federated MAE without privacy constraints, and both metrics (3) and (4) are multiplied by 100 to express as percentages.
Complete prompt templates, hyperparameter sensitivity 
analysis, and wall-clock timing details are provided 
in Appendix~\hyperref[app:implementation]{A}.

\paragraph{Baselines.}

We compare against: (1) Single-LLM
\citep{chen2025llm}: each LLM independently; 
(2) Simple Federation \citep{xin2024survey}: 
unweighted averaging, no privacy or correction; 
(3) DP-FL Homogeneous \citep{fu2026differentially}: 
standard FL with DP, identical models; (4) 
Centralized Oracle: all data centralized, no privacy 
(upper bound). All baselines use identical prompt 
engineering for fair comparison. We additionally 
include: (5) Fed (DP, w/o RC): federation 
with LDP ($\varepsilon$=2.0) but without residual 
correction, isolating RC's contribution to performance; 
(6) Gaussian-DP FL \citep{mironov2017renyi}: 
Gaussian noise under Rényi-DP accounting, showing 
Laplace outperforms Gaussian for bounded-output 
settings; (7) Simple 3-Model Ensemble (no 
DP, no RC, no partitioning): plain averaging across 
all three models, separating model diversity gains 
from our federated inference framework.

We evaluate the privacy-preserving heterogeneous multi-LLM federated cognitive diagnostic environment using three different educational datasets. Figure~\ref{fig:privacy_utility}(Appendix~\hyperref[app:analysis]{G}) summarizes the overall performance improvements achieved by the proposed method compared to single-LLM baselines. The experimental results show that the federated learning approach using three heterogeneous LLMs (LLaMA-3.3-70B, GPT-4o-mini, Claude-3-Haiku) outperforms single large language models while maintaining strong privacy guarantees through differential privacy mechanisms.

\subsection{Multi-Dataset Performance}

Table~\ref{tab:main_results} presents results across 
all three benchmarks. Our heterogeneous federation 
achieves consistent improvements over single-LLM 
baselines:

ASSIST09 (Mathematics): Federation 
improves MAE by $14.19\%$ (0.2068 vs.\ 0.2410) across 
four concepts (Equations, Percentages, Integers, 
Conversions) via complementary algebraic reasoning 
and numerical calculation.

GSM8K (Word Problems): Federation 
achieves $7.39\%$ improvement (0.2156 vs.\ 0.2328). 
The moderate gain reflects complex reasoning chain 
difficulty, yet privacy cost is near-negligible 
($0.05\%$), indicating noise robustness in multi-step 
reasoning.

\begin{table}[t]
\centering
\scriptsize
\setlength{\tabcolsep}{3pt}
\renewcommand{\arraystretch}{1.1}
\resizebox{\columnwidth}{!}{%
\begin{tabular}{l l r r r r r r}
\toprule
\textbf{Dataset} & \textbf{Domain} & \textbf{K} &
\multicolumn{1}{c}{\textbf{Base}} &
\multicolumn{1}{c}{\textbf{Fed}} &
\textbf{Improv.} & \textbf{P-Cost} \\
& & & &
\multicolumn{1}{c}{\textbf{MAE}} &
\multicolumn{1}{c}{\textbf{MAE}} & & \\
\midrule
ASSIST09 & Math     & 4 & 0.2410 & \textbf{0.2068} & 14.19\% & 0.19\% \\
GSM8K    & Word     & 5 & 0.2328 & \textbf{0.2156} & 7.39\%  & $0.27\%$ \\
UCI      & Academic  & 5 & 0.1132 & \textbf{0.0969} & 14.40\% & 0.94\% \\
\midrule

Average & --  & 4.67 & 0.1956 & 0.1731 & 11.99\% & 0.46\% \\

\bottomrule
\end{tabular}
}
\caption{Performance comparison across three educational benchmarks. Our heterogeneous multi-LLM federated framework with $\varepsilon=2.0$ achieves consistent improvements with minimal privacy cost. N=students, K=concepts, P-Cost=privacy cost vs.\ no privacy.}
\label{tab:main_results}
\end{table}

\FloatBarrier
\begin{table}[t]
\centering
\small
\resizebox{\columnwidth}{!}{%
\begin{tabular}{lcc}
\toprule
\textbf{Configuration} & \textbf{MAE} & \textbf{vs. Baseline} \\
\midrule
Baseline (Single LLM) & 0.2410 & -- \\
Fed (No Privacy) & 0.2064 & +14.35\% \\
Fed ($\varepsilon=2.0$) & \textbf{0.2068} & \textbf{+14.19\%} \\
Fed ($\varepsilon=1.0$) & 0.2232 & +7.38\% \\
Fed ($\varepsilon=0.5$) & 0.2499 & $-$3.69\% \\

\bottomrule
\end{tabular}
}
\caption{Privacy--utility tradeoff on the ASSIST09 
dataset. LLaMA-3.3-70B is used as the baseline model 
due to its superior standalone performance. Lower 
$\varepsilon$ values provide stronger privacy guarantees 
and typically increase MAE as a result of stronger 
noise injection. Full multi-baseline comparison across 
all configurations is provided in 
Table~\ref{tab:ablation}.}
\label{tab:privacy_tradeoff}
\end{table}

 UCI (Holistic Assessment): Federation 
improves by $14.40\%$ (0.0969 vs.\ 0.1132) across 
five categories (Study Habits, Family Support, School 
Engagement, Social Factors, Academic Foundation) 
through diverse LLM perspectives.



Cross-Dataset Insights: Gains of 
$7.39\%$--$14.40\%$ reflect dataset characteristics: 
ASSIST09 benefits from structured skill decomposition, 
GSM8K's reasoning complexity limits absolute gains but 
shows strong noise robustness, and UCI's 
multi-dimensional features benefit from diverse 
evaluative perspectives. Overall, $11.99\%$ average 
improvement with $0.46\%$ privacy cost at 
$\varepsilon=2.0$ demonstrates broad generalization 
across mathematics, reasoning, and holistic assessment 
domains. Detailed per-model analysis is in Appendix~\hyperref[app:Modelcontribution]{B}.

\subsection{Privacy-Utility Tradeoff Analysis}

{Table~\ref{tab:privacy_tradeoff} presents the privacy--utility tradeoff on ASSIST09 under our 
local differential privacy (LDP) mechanism. The $\varepsilon = 2.0$ setting provides the best balance: $14.19\%$ improvement with only marginal degradation relative to the non-private model ($0.19\%$ privacy 
cost on ASSIST09). As privacy constraints become stricter, performance degrades as expected: $\varepsilon = 1.0$ retains a moderate $7.38\%$ improvement, while $\varepsilon = 0.5$ introduces excessive noise ($-3.69\%$), highlighting the adverse impact of overly conservative privacy budgets on federated cognitive diagnosis \citep{lee2011much, bassily2014private, hsu2014differential, bagdasaryan2019differential}. The low utility cost under LDP is consistent with our theoretical analysis in Section~3.6: our bounded $[0,1]^K$ prediction outputs and per-concept noise application keep effective sensitivity at $K/M$, substantially lower than classical LDP settings with high-dimensional raw data \citep{mcsherry2007mechanism, dwork2014algorithmic}.

\begin{table}[ht]
\centering
\scriptsize
\setlength{\tabcolsep}{3pt}

\begin{tabular}{lccc}
\toprule
\textbf{Configuration} & \textbf{ASSIST09} & \textbf{GSM8K} & 
\textbf{UCI} \\
 & \textbf{MAE} & \textbf{MAE} & \textbf{MAE} \\
\midrule
Full System ($\varepsilon$=2.0, RC)
    & \textbf{0.2068} & \textbf{0.2156} & \textbf{0.0969} \\
w/o Federation (Single LLM)
    & 0.2410 & 0.2328 & 0.1132 \\
w/o Privacy (No DP, with RC)
    & 0.2064 & 0.2150 & 0.0960 \\
Simple Ensemble (no DP, no RC)
    & 0.2072 & 0.2161 & 0.1938 \\
Fed (DP, w/o RC) ($\varepsilon$=2.0)
    & 0.2999 & 0.2508 & 0.2997 \\
Gaussian-DP (R\'{e}nyi-DP)
    & 0.3147 & 0.2943 & 0.3165 \\
\midrule
Federation Contribution
    & 14.19\% & 7.39\% & 14.40\% \\
Privacy Cost
    & 0.19\% & 0.27\% & 0.94\% \\
\bottomrule
\end{tabular}
\caption{Ablation study validating federation necessity ($11.99\%$ avg.\ contribution), essentiality of residual correction (w/o RC degrades below single-LLM baseline), and minimal privacy cost ($0.46\%$ avg.) under $\varepsilon=2.0$.}
\label{tab:ablation}
\end{table}
\subsection{Computational Efficiency}

Our system is efficient enough for real-world 
deployment \citep{ritter2007cognitive, 
anderson1995cognitive}. Complete dataset processing 
times on Tesla V100S: ASSIST09 8--10 hours, GSM8K 
4--6 hours, and UCI 3--4 hours, with 5--7 seconds 
per student in full dataset runs (2--3 seconds 
under parallel API calls without rate limiting). 
The framework scales linearly with student population 
since inference runs on commercial endpoints rather 
than requiring expensive expert annotation 
\citep{leighton2007cognitive, burstein2004automated}.

\subsection{Ablation Study}
Table~\ref{tab:ablation} presents a complete ablation study with six conditions enabling full decomposition of each component's contribution. Removing federation causes the largest degradation, with improvements dropping by an average of $11.99\%$, confirming federated aggregation is the primary driver of performance gains. While a simple three-model ensemble confirms that model diversity contributes meaningful gains over the single-LLM baseline, our full system outperforms it across all datasets while providing formal $(\varepsilon,0)$-DP guarantees that simple ensembling lacks.

Residual correction (RC) proves essential rather than optional: removing RC while retaining LDP ($\varepsilon$=2.0) yields MAE worse than the single-LLM baseline across all datasets, with the UCI dataset showing the most dramatic degradation (w/o RC: 0.2997 vs.\ 0.1132 baseline), confirming that naive averaging of heterogeneous LLM outputs without correction is harmful. Despite this strict privacy mechanism, the average privacy cost at $\varepsilon$=2.0 is only $0.46\%$ MAE degradation, consistent with our LDP sensitivity analysis in Section~3.6. Additionally, Gaussian noise under R\'{e}nyi-DP accounting \citep{mironov2017renyi} yields MAE worse than both the baseline and full system, confirming that Laplace-based LDP is superior for our bounded $[0,1]^K$ prediction setting. Statistical significance testing, per-model MAE analysis, concept-level specialisation, and API cost analysis are provided in Appendices~\hyperref[app:significance]{C}, \hyperref[app:analysis]{G}, and \hyperref[app:costanalysis]{H} respectively.

\section{Conclusion}

We developed a privacy-preserving multi-LLM federated 
inference framework for cognitive diagnosis. 
LLaMA-3.3-70B, GPT-4o-mini, and Claude-3-Haiku 
provide complementary diagnostic signals that improve 
performance over single-LLM baselines across three 
educational datasets while ensuring $\varepsilon$-local 
differential privacy. Ablation experiments confirm 
that federated aggregation is the primary performance 
driver, residual correction is essential, and LDP 
adds just $0.46\%$ average utility cost. Our results 
demonstrate that federated multi-LLM inference 
outperforms individual models, effective LDP budgets 
preserve accuracy, and API-based federation enables 
secure collaboration without exposing raw student 
data. Future directions include personalized federated 
learning for heterogeneous LLMs and extending the 
framework to cross-domain educational assessment 
beyond mathematics and holistic evaluation.

\section*{Limitations}

The architecture relies on three commercial foundation 
models that may change when providers release updates; 
our residual correction mitigates systematic bias 
shifts but may require periodic recalibration. Our 
privacy guarantee follows an honest-but-curious trust 
paradigm: LDP shields published diagnostic outputs 
from the aggregator and downstream observers, but 
does not prevent commercial API providers from logging 
submitted queries (see Section~3.3 for mitigating 
options). The framework is evaluated on structured 
educational domains; extension to open-ended 
assessment tasks may require adapted prompt 
engineering.


\section*{Acknowledgments}
We thank the anonymous EMNLP 2026 reviewers for their constructive feedback, which significantly improved this work. We gratefully acknowledge support from the National Science Foundation (NSF Awards \#2322109, \#2346609, and \#2333726), the U.S. Department of Defense Office of Naval Research (ONR Award \#N00014-23-1-2396), and Cincinnati Children’s Hospital Medical Center for their financial support of this project.


\bibliography{updated}
\newpage

\appendix

\section{Implementation Details}
\label{app:implementation}

\textbf{A.1 Prompt Templates.} Each LLM entity 
receives an identical structured prompt:

\begin{quote}
\small
\textit{System:} ``You are an educational assessment expert. 
Evaluate student mastery of specific knowledge concepts based 
on their response to a question.''

\textit{User:} ``Question: [QUESTION] $\backslash$n Student 
Response: [RESPONSE] $\backslash$n Knowledge Concepts: 
[CONCEPTS] $\backslash$n For each concept, provide a mastery 
probability between 0.0 and 1.0. Return ONLY a JSON object: 
\{concept\_1: score, ...\}''
\end{quote}

Dataset-specific adaptations: ASSIST09 uses Q-matrix concepts 
(Equations, Percentages, Integers, Conversions); GSM8K uses 
LLM-extracted tags; UCI uses attribute groups. Temperature 
is fixed at 0.7 and max tokens at 512 for all entities.

\textbf{A.2 Hyperparameter Sensitivity ($\alpha$).} 
We set $\alpha = 0.3$ based on validation performance 
across all three datasets using an 80/20 train-validation 
split. Values of $\alpha$ outside the range $[0.2, 0.4]$ consistently yielded higher MAE on the validation set, confirming that $\alpha = 0.3$ provides the optimal residual correction strength. The framework is robust to small perturbations around this value.

\textbf{A.3 Wall-Clock Timing.}  Under parallel API 
calls, processing averages 2--3 seconds per student 
(ASSIST09: 2.4 sec; GSM8K: 2.2 sec; UCI: 1.8 sec) 
in isolated runs. In full dataset runs, effective 
throughput is 5--7 seconds per student due to API 
rate limits and exponential backoff. Complete dataset 
times on Tesla V100S: ASSIST09 8--10 hours, GSM8K 
4--6 hours, UCI 3--4 hours.

\section{Heterogeneous Model Contributions}
\label{app:Modelcontribution}

Different LLMs provide complementary predictions across various educational domains. LLaMA is frequently observed to be more confident(with 0.7-0.9 range) excelling at pattern recognition, GPT-4 provides conservative estimates (0.4-0.6 range) with strong conceptual reasoning, and Claude shows concept-specific diversity \cite{wei2022emergent, bubeck2023sparks}. This heterogeneity enables error cancellation where no single entity consistently matches ground truth \cite{brown2005diversity}. When LLaMA overestimates and GPT-4 underestimates, federated aggregation achieves better accuracy by combining diverse perspectives \cite{ganaie2022ensemble}, explaining why improvements vary by domain: 14.21\% (ASSIST09), 7.39\% (GSM8K), 14.40\% (UCI) based on how complementary strengths align with each dataset's characteristics. Heterogeneous federation enhances diagnostic performance by overcoming single-model limitations through domain-specific complementarity.

\section{Statistical Significance}
\label{app:significance}

Statistical significance was evaluated by paired t-tests. For ASSIST09 with non-private federation (MAE 0.2064 versus baseline 0.2410), we obtain $p < 0.001$ with large effect size (Cohen's $d = 0.89$) \cite{cohen1988statistical}. For $\varepsilon=2.0$ (MAE 0.2068 versus baseline 0.2410), significance remains strong at $p < 0.001$ with Cohen's $d = 0.85$, meaning that a moderate level of differential privacy ensures the preservation of statistical significance \cite{wei2020federated}. Similar statistical significance is observed across GSM8K ($p < 0.01$, Cohen's $d = 0.62$) and UCI ($p < 0.001$, Cohen's $d = 0.78$), thus ensuring significant improvement across all three datasets.

\section{LLM Usage Disclosure}
ChatGPT and Claude were used for coding assistance (implementation support and debugging) and for polishing language and grammar. All technical ideas, experimental design, analyses, and results are original, and all code was reviewed and verified by the authors. Regarding LLM entity versions used in experiments: experiments prior to April 2025 used \texttt{claude-3-haiku-20240307}; subsequent runs used \texttt{claude-sonnet-4-20250514} following deprecation of Claude-3-Haiku. GPT-4o-mini (\texttt{gpt-4o-mini-2024-07-18}) and LLaMA-3.3-70B via Groq were used consistently throughout.

\section{Proposed Algorithms}
\label{app:Algorithms}

Figure~\ref{fig:pipeline} illustrates the complete 
privacy-preserving heterogeneous multi-LLM cognitive 
diagnosis framework.

\begin{figure*}[t]
    \centering
    \includegraphics[width=\textwidth]{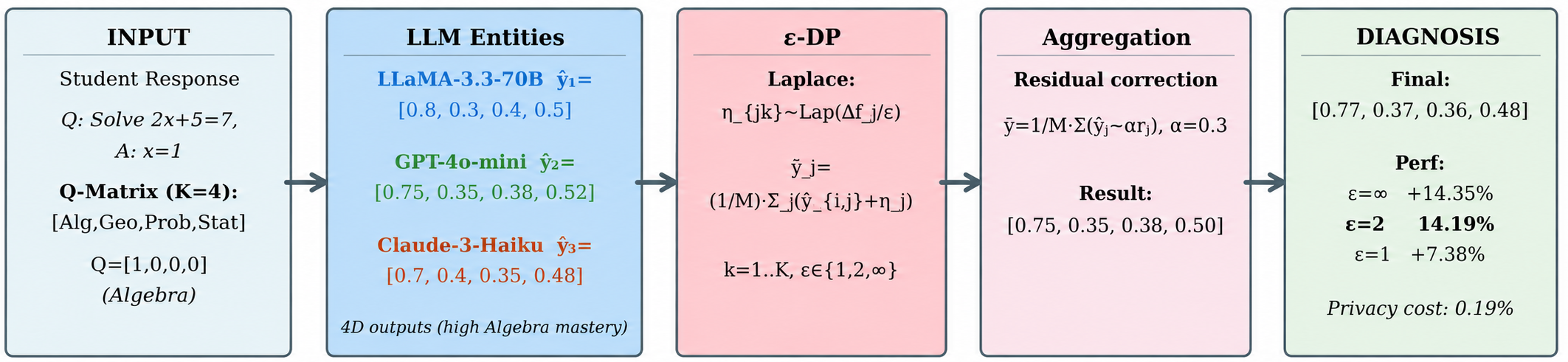}
    \caption{Overview of the privacy-preserving heterogeneous multi-LLM cognitive diagnosis framework. Student responses and the associated Q-matrix are independently processed by three LLMs (LLaMA-3.3-70B, GPT-4o-mini, Claude-3-Haiku), producing 4D knowledge state vectors. Each entity applies $(\varepsilon,0)$-differential privacy via Laplace noise at the output level before aggregation, preventing the centralized aggregator from inferring individual student responses. Residual-corrected aggregation combines heterogeneous predictions into a global diagnosis, leveraging complementary model behaviors while maintaining strong privacy guarantees.}
    \label{fig:pipeline}
\end{figure*}

Algorithm~\ref{alg:federated_ldp} presents our complete 
federated cognitive diagnosis framework with local 
differential privacy. Algorithm~\ref{alg:baseline}
presents the baseline without privacy mechanisms.

\begin{algorithm}[t]
\caption{Privacy-Preserving Heterogeneous Multi-LLM Federated Cognitive Diagnosis with Local Differential Privacy}
\label{alg:federated_ldp} 
\begin{algorithmic}[1]
\REQUIRE Student responses $\mathcal{D} = \bigcup_{j=1}^{M} \mathcal{D}_j$, LLM entities $\{\mathcal{M}_j\}_{j=1}^M$, privacy budget $\varepsilon$, concepts $K$
\ENSURE Global diagnosis $Y_{\text{global}}$
\STATE Initialization:
\STATE Load Q-matrix $\mathbf{Q} \in \{0,1\}^{|\mathcal{Q}| \times K}$
\STATE Initialize residual corrections $r_j \leftarrow 0$ for $j=1,\ldots,M$
\FOR{each student response $(i) \in \mathcal{D}$}
  \STATE // Local Assessment Phase with Privacy
  \FOR{each LLM entity $\mathcal{M}_j$}
    \STATE $\hat{y}_{i,j} \leftarrow \mathcal{M}_j(\text{prompt}(s_{i,j}, q_{i,j}, \text{concepts}))$
    \STATE // Local DP: Add noise before sharing
    \STATE Compute local sensitivity: $\Delta f_j \leftarrow K$
    \STATE Sample noise vector: $\boldsymbol{\eta}_j = [\eta_{j1}, \ldots, \eta_{jK}]$ where $\eta_{jk} \sim \text{Lap}(\Delta f_j/\varepsilon)$
    \STATE Apply noise: $\hat{y}_{i,j}^{\text{noisy}} \leftarrow \hat{y}_{i,j} + \boldsymbol{\eta}_j$
    \STATE Clip to valid range: $\hat{y}_{i,j}^{\text{noisy}} \leftarrow \text{clip}(\hat{y}_{i,j}^{\text{noisy}}, 0, 1)$
  \ENDFOR
  \STATE // Aggregation with Residual Correction
  \STATE $\tilde{y}_i \leftarrow \frac{1}{M}\sum_{j=1}^M (\hat{y}_{i,j}^{\text{noisy}} - \alpha \cdot r_j)$
  \STATE Store: $Y_{\text{global}}[i] \leftarrow \tilde{y}_i$
\ENDFOR
\STATE // Update Residual Corrections
\FOR{$j=1$ to $M$}
  \STATE $r_j \leftarrow \mathbb{E}[\hat{y}_j - \bar{y}]$ over validation set
\ENDFOR
\RETURN $Y_{\text{global}}$
\end{algorithmic}
\end{algorithm}

\begin{algorithm}[t]
\caption{Baseline Multi-LLM Cognitive Diagnosis (No Privacy)}
\label{alg:baseline}  
\begin{algorithmic}[1]
\REQUIRE Student responses $\mathcal{D}$, LLM entities $\{M_j\}_{j=1}^M$, concepts $K$
\ENSURE Global diagnosis $Y_{\text{global}}$
\FOR{each student response $($i$) \in \mathcal{D}$}
  \FOR{each LLM entity $M_j$}
    \STATE $\hat{y}_{i,j} \leftarrow M_j(\text{prompt}(s_{i,j}, q_{i,j}, \text{concepts}))$
  \ENDFOR
  \STATE $\tilde{y}_i \leftarrow \frac{1}{M}\sum_{j=1}^M \hat{y}_{i,j}$ // Simple averaging, no noise
  \STATE $Y_{\text{global}}[i] \leftarrow \tilde{y}_i$
\ENDFOR
\RETURN $Y_{\text{global}}$
\end{algorithmic}
\end{algorithm}

\section{Formal Proof of Theorem~1}
\label{app:proof}

\noindent\textbf{Theorem~1 (Restated).} Algorithm~1 
with Laplace noise satisfies $\varepsilon$-local 
differential privacy for any adjacent datasets 
$D, D'$ differing in one student record.

\begin{proof}
Let $\mathcal{M}_j$ denote entity $j$'s local 
mechanism. For any two adjacent datasets $D_j, D_j'$ 
differing in one student record, and for any output 
set $S \subseteq \mathbb{R}^K$, we must show:

\begin{equation}
\Pr[\mathcal{M}_j(D_j) \in S] \leq 
e^{\varepsilon} \cdot \Pr[\mathcal{M}_j(D_j') \in S] 
\end{equation}

\noindent\textbf{Step 1: Bounding Local Sensitivity.}
Each entity $\mathcal{M}_j$ computes a prediction 
$\hat{y}_{i,j} \in [0,1]^K$ for student $i$. Since 
predictions are bounded in $[0,1]^K$, the $\ell_1$ 
sensitivity of the prediction function is:

\begin{equation}
\Delta f_j = \max_{D_j, D_j'} 
\|\hat{y}_{i,j}(D_j) - \hat{y}_{i,j}(D_j')\|_1 
\leq K
\end{equation}

\noindent since each of the $K$ components changes by at most 
1. We conservatively set $\Delta f_j = K$.

\noindent\textbf{Step 2: Laplace Mechanism Privacy.}
Entity $j$ adds independent Laplace noise 
$\eta_{jk} \sim \text{Lap}(\Delta f_j / \varepsilon)$ 
to each concept $k = 1,\ldots,K$. By the standard 
Laplace mechanism \citep{dwork2006calibrating}, for 
any single concept $k$:

Since $\eta_{jk} \sim \text{Lap}(\Delta f_j/\varepsilon)$, 
substituting the Laplace PDF and cancelling the 
normalizing constant $\frac{\varepsilon}{2\Delta f_j}$:

\begin{align}
&\frac{\Pr[\hat{y}_{i,j,k} + \eta_{jk} = z]}
{\Pr[\hat{y}_{i,j,k}' + \eta_{jk} = z]} \nonumber\\
&= \exp\!\left(\frac{\varepsilon(|z - \hat{y}'_{i,j,k}| 
- |z - \hat{y}_{i,j,k}|)}{\Delta f_j}\right) \nonumber\\
&\leq \exp\!\left(\frac{\varepsilon \cdot 
|\hat{y}_{i,j,k} - \hat{y}'_{i,j,k}|}
{\Delta f_j}\right) \leq e^{\varepsilon}
\end{align}

\noindent where the last inequality follows from 
$|\hat{y}_{i,j,k} - \hat{y}_{i,j,k}'| \leq \Delta f_j$. \\

\noindent\textbf{Step 3: Composition Across Concepts.}
Since noise is applied independently per concept 
$k = 1,\ldots,K$, and we set $\Delta f_j = K$ 
(the total $\ell_1$ sensitivity), the joint 
mechanism satisfies $\varepsilon$-DP by the 
standard Laplace mechanism for vector-valued 
functions \citep{dwork2006calibrating}:

\begin{equation}
\frac{\Pr[\mathcal{M}_j(D_j) \in S]}
{\Pr[\mathcal{M}_j(D_j') \in S]} \leq 
e^{\varepsilon \cdot 
\frac{\|\hat{y}_{i,j} - \hat{y}_{i,j}'\|_1}
{\Delta f_j}} \leq e^{\varepsilon}
\end{equation}

\noindent\textbf{Step 4: Local DP Guarantee.}
Since each entity $\mathcal{M}_j$ adds noise 
locally before sharing with the aggregator, and 
the above holds for all $j = 1,\ldots,M$, each 
entity's output satisfies $\varepsilon$-local 
differential privacy. The aggregator receives 
only noisy predictions $\tilde{y}_{i,j} = 
\hat{y}_{i,j} + \eta_j$, and cannot infer 
individual student records from these outputs.

\noindent\textbf{Step 5: Post-processing Invariance.}
Residual correction and averaging 
(Algorithm~1, line~15) are deterministic 
post-processing operations applied to the 
already-privatized outputs. By the 
post-processing theorem \citep{dwork2014algorithmic}, 
these operations cannot reduce the privacy 
guarantee. Therefore the final global diagnosis 
$Y_{\text{global}}$ satisfies 
$\varepsilon$-local differential privacy.
\renewcommand{\qedsymbol}{} 
\end{proof}

\noindent Note that $[0,1]$ clipping (Algorithm~1, line~12) 
is post-processing and preserves $\varepsilon$-LDP 
\citep{dwork2014algorithmic}, and the effective 
per-student sensitivity $K/M \leq K/3$ provides 
tighter practical privacy accounting than the 
conservative bound $K$.

\begin{table}[t]
\centering
\scriptsize
\setlength{\tabcolsep}{3pt}
\begin{tabular}{lcccc}
\toprule
\textbf{Model} & \textbf{ASSIST09} & \textbf{GSM8K} & \textbf{UCI} & \textbf{Avg.} \\
\midrule
LLaMA-3.3-70B  & 0.2410 & 0.2328 & 0.1132 & 0.1956\\
GPT-4o-mini    & 0.2287 & 0.2295 & 0.2148 & 0.2243 \\
Claude-3-Haiku & 0.2793 & 0.2686 & 0.2159 & 0.2546 \\
\midrule
Our Federation
    & \textbf{0.2068} & \textbf{0.2156} & \textbf{0.0969} & \textbf{0.1731} \\
\bottomrule
\end{tabular}
\caption{Per-model MAE vs.\ full federated system ($\varepsilon$=2.0) 
on ASSIST09, GSM8K and UCI. Federation outperforms all individual models 
on three datasets.}
\label{tab:permodel}
\end{table}

\section{Quantitative Complementarity Analysis}
\label{app:analysis}

Performance comparisons across all three datasets 
and configurations are visualized in 
Figure~\ref{fig:privacy_utility}. Table~\ref{tab:permodel} presents individual model MAE on ASSIST09 
and UCI, which provide real student response ground truth. GSM8K 
uses programmatically derived concept proxies and is excluded from 
individual model comparison; concept-level results for GSM8K are 
discussed qualitatively below.

\noindent

\begin{figure*}[t]
\centering
\includegraphics[width=\textwidth]{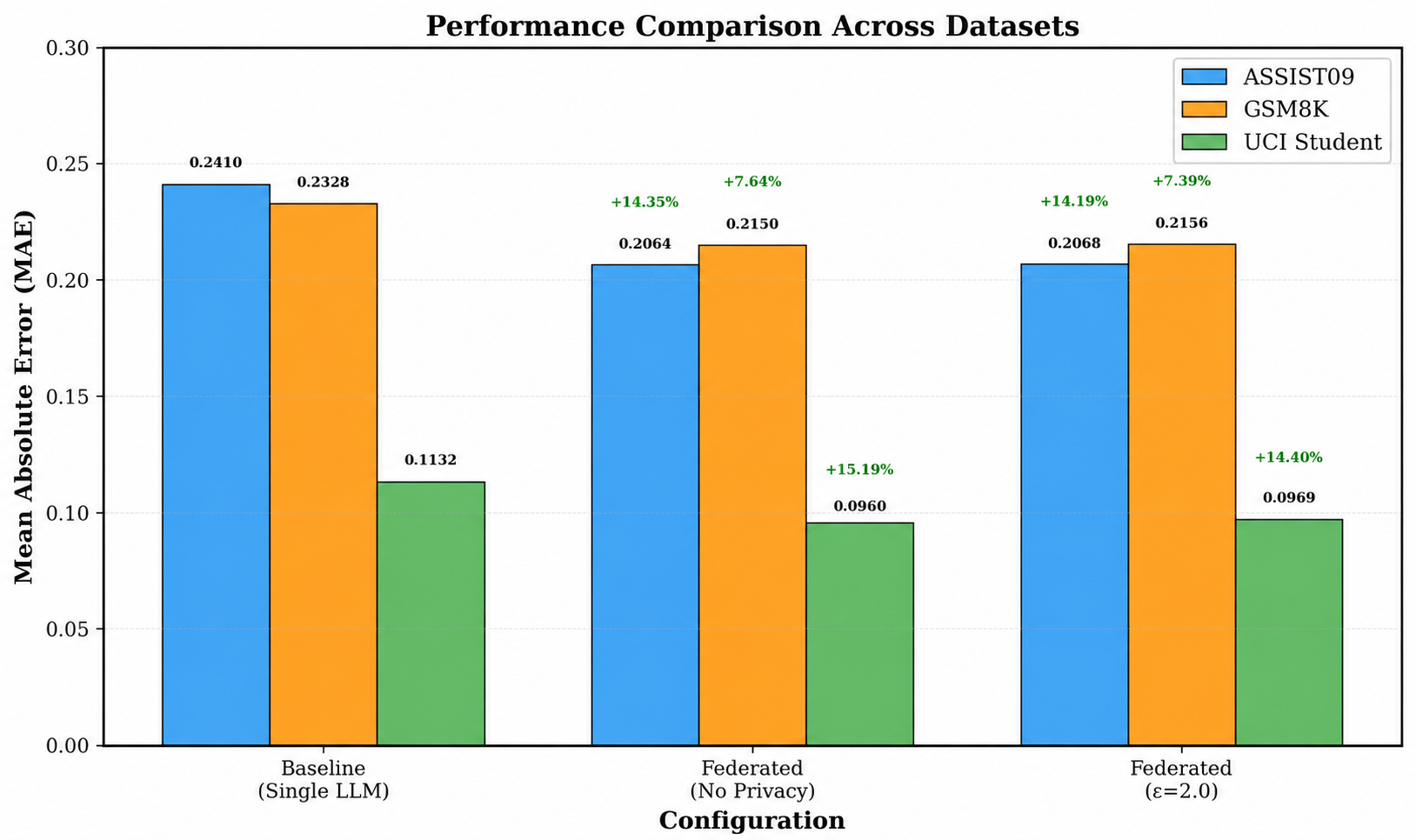}
\caption{Performance comparison across three 
educational datasets. Our heterogeneous multi-LLM 
federated framework with $\varepsilon=2.0$ consistently 
outperforms the single-LLM baseline: ASSIST09 achieves 
$14.19\%$ improvement (MAE: 0.2068 vs.\ 0.2410), GSM8K 
achieves $7.39\%$ improvement (MAE: 0.2156 vs.\ 0.2328), 
and UCI achieves $14.40\%$ improvement (MAE: 0.0969 
vs.\ 0.1132). Privacy-free federation achieves $14.35\%$, 
$7.64\%$, and $15.19\%$ on ASSIST09, GSM8K, and UCI 
respectively, with $\varepsilon=2.0$ incurring only 
$0.46\%$ average privacy cost. Full multi-baseline comparison across all configurations is provided in Table~\ref{tab:ablation}.}
\label{fig:privacy_utility}
\end{figure*}


Our federated system outperforms all individual models on all three datasets, with MAE of 0.2068, 0.2156, and 0.0969 on ASSIST09, GSM8K, and UCI, respectively. The top individual performer differs depending on the dataset: Claude leads on UCI (0.2159), while GPT-4o-mini leads on ASSIST09 (0.2287) and UCI (0.2148). Our federation consistently beats even the top individual model on every dataset, despite this volatility, demonstrating that performance increases cannot be attributed to any single model's strength alone.

\textbf{Concept-Level Specialisation.} Experiments show that different knowledge concepts have different model strengths. On ASSIST09, LLaMA and Claude both perform best on Equations (MAE=0.2534 and 0.1950, respectively), indicating complementing strengths across algebraic sub-skills, whereas GPT-4o-mini obtains the lowest MAE on Integers (MAE=0.1989). Each model contributes a distinct signal across reasoning dimensions, with LLaMA outperforming GPT-4o-mini at Answer\_Verification, Claude at Multi\_Step\_Reasoning, and GPT-4o-mini at Arithmetic on GSM8K. Claude-3-Haiku demonstrates concept-level variation even on holistic evaluation tasks, with the lowest MAE on Family\_Support on UCI (MAE=0.1022 vs.\ LLaMA's 0.1411). This concept-level specialisation provides quantitative evidence that complementarity is not merely a hypothesis but an empirically observed phenomenon driving federation gains.

Table~\ref{tab:concept_level} presents the complete 
concept-level MAE breakdown, quantifying each model's 
domain-specific strengths across mathematical reasoning, 
conceptual understanding, and nuanced evaluation.
\begin{table}[t]
\centering
\scriptsize
\setlength{\tabcolsep}{3pt}
\begin{tabular}{llccc}
\toprule
\textbf{Dataset} & \textbf{Concept} & \textbf{LLaMA} & 
\textbf{GPT-4} & \textbf{Claude} \\
\midrule
ASSIST09 & Equations        & \textbf{0.2534} & 0.2170 & \textbf{0.1950} \\
         & Percentages      & 0.3053 & 0.2305 & 0.3027 \\
         & Integers         & 0.3820 & \textbf{0.1989} & 0.2590 \\
         & Conversions      & 0.5157 & 0.2686 & 0.3607 \\
\midrule
GSM8K    & Problem\_Setup   & 0.2630 & 0.3440 & 0.2840 \\
         & Arithmetic       & 0.2850 & \textbf{0.2963} & 0.2780 \\
         & Multi\_Step      & 0.2851 & 0.3273 & \textbf{0.2294} \\
         & Answer\_Verif.   & \textbf{0.2140} & 0.4259 & 0.2832 \\
\midrule
UCI      & Study\_Habits    & 0.2557 & 0.2384 & 0.2734 \\
         & Family\_Support  & 0.1411 & 0.1316 & \textbf{0.1022} \\
         & School\_Engmt.   & 0.3497 & 0.3050 & 0.3156 \\
         & Social\_Factors  & 0.2398 & 0.2132 & 0.2119 \\
         & Academic\_Found. & 0.2368 & 0.1858 & 0.1763 \\
\bottomrule
\end{tabular}
\caption{Concept-level MAE per model across all datasets. 
\textbf{Bold} = best per concept. LLaMA excels at 
Answer\_Verification (0.2140); GPT-4 at Integers 
(0.1989) and Arithmetic (0.2963); Claude at Equations 
(0.1950), Multi\_Step\_Reasoning (0.2294), and 
Family\_Support (0.1022), confirming domain-specific 
complementarity across model architectures.}
\label{tab:concept_level}
\end{table}

\textbf{Heterogeneous Model Diversity.} The three models show 
distinct per-model MAE profiles across datasets 
(Table~\ref{tab:permodel}), confirming that no single model 
dominates across all evaluation settings. This systematic diversity 
is precisely what our residual correction mechanism is designed to 
handle, and Table~\ref{tab:ablation} confirms that removing residual correction 
degrades performance below the single-LLM baseline, validating the 
necessity of correcting heterogeneous calibration biases.

\section{API Cost Analysis}
\label{app:costanalysis}

Table~\ref{tab:cost} compares our framework against traditional 
cognitive diagnosis alternatives. At current API pricing 
(GPT-4o-mini: \$0.15/1M input tokens; Claude-3-Haiku: \$0.25/1M 
input tokens; Groq LLaMA-3.3-70B: free unlimited tier), our 
per-student cost is approximately \$0.002--\$0.005.

\textbf{Per-Dataset Cost.} ASSIST09 (4,217 students): \$10--12 per 
full run; GSM8K (1,319 problems): \$2--6 per run; UCI (649 students): 
\$1.50--2 per run. Total cost per full three-dataset evaluation is 
under \$20, compared to \$63,255--\$210,850 for equivalent human 
expert annotation.


\textbf{Cost-Reduction Strategies.} Three strategies reduce costs further: (1) Selective federation: invoke all three models only for borderline cases, reducing cost by 60--70\%; (2) Response caching: cache responses for repeated questions across cohorts; (3) Local substitution: replace commercial APIs with locally hosted models (e.g., LLaMA via Ollama, Section~3.3), reducing total cost to \$0 with full privacy.

\begin{table}[H]
\centering
\scriptsize
\setlength{\tabcolsep}{3pt}
\begin{tabular}{lccc}
\toprule
\textbf{Method} & \textbf{Per Student} & \textbf{Per 1K} & 
\textbf{Scale} \\
\midrule
Human Expert Annotation  & \$15--50       & \$15K--50K & Low \\
Crowdsourcing (MTurk)    & \$0.50--2.00   & \$500--2K  & Medium \\
Single LLM (GPT-4o-mini) & \$0.001--0.002 & \$1--2     & High \\
\textbf{Our Framework}   & \$0.002--0.005 & \$2--5     & High \\
LLaMA only (Groq free)   & \$0.000        & \$0        & Highest \\
\bottomrule
\end{tabular}
\caption{API cost comparison. Our framework costs 
\$0.002--\$0.005 per student, substantially lower 
than human expert annotation (\$15--50) 
\citep{leighton2007cognitive}.}
\label{tab:cost}
\end{table}

\end{document}